\documentclass[pdflatex,sn-mathphys-num]{sn-jnl}
\usepackage{graphicx}
\usepackage{mathtools}
\usepackage[linesnumbered,ruled,vlined,norelsize]{algorithm2e}

\usepackage{tabularx}
\usepackage{environ}

\usepackage{url}

\usepackage[caption=false]{subfig}

\usepackage{enumerate}

\usepackage{multirow}%
\usepackage{amsfonts}%
\usepackage{amsthm}%
\usepackage{cleveref}
\Crefname{figure}{{Fig.}}{{Figs.}}
\usepackage{booktabs}

\usepackage{etoolbox}
\apptocmd{\thebibliography}{\raggedright}{}{}

\newtheorem{lemma}{Lemma}
\newtheorem{theorem}{Theorem}
\newtheorem{corollary}[lemma]{Corollary}
\newtheorem{observation}[lemma]{Observation}

\newtheorem{problem}{Problem}

\newcommand{\set}[1]{\{#1\}}
\newcommand{\inset}[2]{\{#1 \mid #2\}}
\newcommand{\size}[1]{\left|#1\right|}
\newcommand{\order}[1]{O(#1)}

\newcommand{\In}{\mathit{In}}
\newcommand{\Ex}{\mathit{Ex}}
\newcommand{\Vin}{V_\mathit{In}}

\makeatletter
\@ifpackageloaded{algorithm2e}{
\SetKwProg{Fn}{Function}{}{}
\SetKwProg{Procedure}{Procedure}{}{}
\SetKwComment{tcc}{//}{}
\SetKwFunction{Output}{Output}%
\SetKwFunction{EnumCIS}{ECIS}
\SetKwFunction{EnumCISImp}{Improved-ECIS}

\SetKwFunction{EnumCS}{ECS}
\SetKwFunction{RecCS}{RecCS-base}
\SetKwFunction{RecCStwo}{RecCS-2}
\SetKwFunction{RecCSthree}{RecCS-3}
\SetKwFunction{RecCSfour}{RecCS-4}
\SetKwFunction{EnumCSImp}{Improved-ECS}

\SetKwFunction{EnumEC}{EEC}
\SetKwFunction{RecEC}{RecEC}
\SetKwFunction{EnumEdge}{Enum-E}{}{}

\SetKwFunction{EnumBC}{EBC}
\SetKwFunction{RecBC}{RecBC}

\SetKw{Continue}{continue} 
\SetKwInOut{AlgInput}{Input}
\SetKwInOut{AlgOutput}{Output}
\SetKwInOut{AlgPrecondition}{Pre-conditions}
\SetKwInOut{AlgInvariant}{Invariants}
}{}
\makeatother

\newcommand{\trim}[1]{{\tt trim(}#1{\tt )}}

\newcommand{\Efar}{\ensuremath{E_{\tt far}}\xspace}

\newcommand{\changed}[1]{#1}

\title{Enumerating Graphlets with Amortized Time Complexity Independent of Graph Size}

\author[5]{\fnm{Alessio} \sur{Conte}}\email{alessio.conte@unipi.it}
\author[5]{\fnm{Roberto} \sur{Grossi}}\email{roberto.grossi@unipi.it}
\author[1]{\fnm{Yasuaki} \sur{Kobayashi}}\email{koba@ist.hokudai.ac.jp}
\author[3]{\fnm{Kazuhiro} \sur{Kurita}}\email{kurita@i.nagoya-u.ac.jp}
\author*[6]{\fnm{Davide} \sur{Rucci}}\email{davide.rucci@isti.cnr.it}
\author[4]{\fnm{Takeaki} \sur{Uno}}\email{uno@nii.ac.jp}
\author[2]{\fnm{Kunihiro} \sur{Wasa}}\email{wasa@hosei.ac.jp}

\affil[1]{\orgname{Hokkaido University}, \city{Sapporo}, \country{Japan}}
\affil[2]{\orgname{Hosei University}, \city{Tokyo}, Japan}
\affil[3]{Nagoya University, \city{Nagoya}, \country{Japan}}
\affil[4]{\orgname{National Institute of Informatics}, \city{Tokyo}, \country{Japan}}
\affil[5]{\orgname{University of Pisa}, \city{Pisa}, \country{Italy}}
\affil*[6]{\orgname{ISTI-CNR}, \city{Pisa}, \country{Italy}}

\date{}
\begin{document}

\abstract{
Graphlets of order~$k$ in a graph $G$ are connected subgraphs induced by $k$~nodes (called $k$-graphlets) or by $k$~edges (called edge $k$-graphlets).
They are among the interesting subgraphs in network analysis to get insights on both the local and global structure of a network.
While several algorithms exist for discovering and enumerating graphlets, the \changed{amortized time complexity} of such algorithms typically depends on the size of the graph $G$, or its maximum degree. In real networks, even the latter can be in the order of millions, whereas $k$ is typically required to be a small value.

In this paper we provide the first algorithm to list all graphlets of order~$k$ in a graph $G=(V,E)$ with an \changed{amortized time complexity} depending \emph{solely} on the order $k$, contrarily to previous approaches where the cost depends \emph{also} on the size of $G$ or its maximum degree. Specifically, we show that it is possible to list $k$-graphlets in $O(k^2)$ time per solution, and to list edge $k$-graphlets in $O(k)$ time per solution. 
Furthermore we show that, if the input graph has bounded degree, then the \changed{amortized time} for listing $k$-graphlets is reduced to $O(k)$. Whenever $k = O(1)$, as it is often the case in practical settings, these algorithms are the first to achieve constant time per solution.
}


\maketitle

\section{Introduction}

Subgraphs are a well-known example of \emph{patterns} or \emph{motifs} that can be extracted from graphs to discover meaningful insights on both the local and global structure of the network they represent. In recent years there has been growing attention on the graph motif problem~\cite{lacroix2006motif, ciriello2008review_motifs, yu2020motif_survey} which asks to list or count a specific kind of pattern (formalized as subgraphs) in a complex network. 
These subgraphs not only offer valuable information on the underlying network's properties, they also serve as building blocks in many fields such as biology and bioinformatics, sociology, computer networks, and transport engineering: in fact, in these kinds of networks, subgraphs are able to capture and highlight the community structure of the network itself \cite{FORTUNATO_community}, a concept that is ever so important when people or any other entities are tightly interconnected. 
In this paper, we study two possible specializations of the main subgraph definition for a given integer, called order $k$: vertex-induced $k$-subgraphs (known as $k$-graphlets~\cite{10.1093/bioinformatics/bth436}) and edge-induced $k$-subgraphs, as illustrated in~\Cref{fig:subgraph_example}. 
Namely, a $k$-\emph{graphlet} is a connected subgraph induced by $k$ vertices.
In this paper, to be consistent with $k$-graphlets, we refer to a connected edge-induced subgraph with $k$ edges as \emph{edge $k$-graphlet}.
\begin{figure}[ht]
    \centering
    \subfloat[]{
        \centering
        \includegraphics[scale=.53]{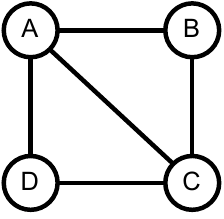}
        \label{fig:subgraph_example:sub:graph}
    }
    \hfill
    \subfloat[]{
        \centering
        \includegraphics[scale=.5]{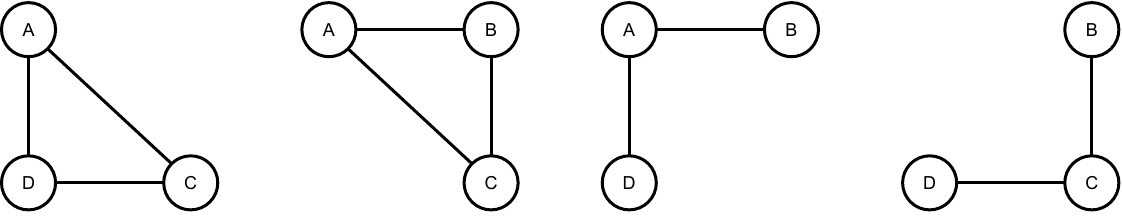}
        \label{fig:subgraph_example:sub:graphlets}
    }
    \vfill
    \subfloat[]{
        \centering
        \includegraphics[width=.7\textwidth]{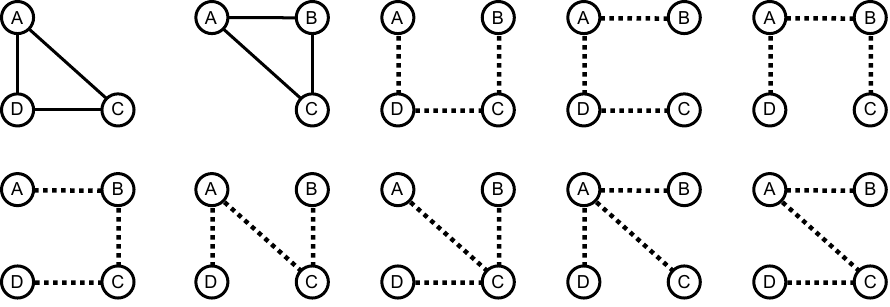}
        \label{fig:subgraph_example:sub:edge-graphlets}
    }
    \caption{(a) A graph $G=(V,E)$ with $|V|=4$ vertices and $|E|=5$ edges. (b) All \emph{3-graphlets} contained in $G$. (c) All edge 3-subgraphs, i.e., \emph{edge $3$-graphlets} in $G$, where dotted edges denote subgraphs that are also acyclic (\emph{3-subtrees}). Note that some of the $k$-graphlets are also edge $k$-graphlets, but this is not necessarily so; also, some edge $k$-graphlets have a number of vertices different from $k$. }
    \label{fig:subgraph_example}
\end{figure}

We observe that $k$-graphlets are getting momentum for their numerous fields of application, such as large graph comparison~\cite{graphlet_kernels, Kriege2020graphkernels} again by kernelization, understanding the role of genes in pathways and cancer mechanisms~\cite{windels2022graphlet}, and finding important nodes in a network~\cite{aparicio2019finding}. Special cases of $k$-graphlets, such as triangles when $k = 3$, or cliques, have been investigated for community detection purposes, as they help speeding up the process of identifying (possibly larger) communities in graphs \cite{klymko2014using, shaping_communities_triangles, triangles_social_cohesion, latent_triangle}. Also, edge $k$-graphlets find their usage in wildly different domains~\cite{chin2018subtrees}, such as graph databases where they are used for indexing purposes~\cite{subtree_graph_indexing}, kernelization in machine learning models~\cite{GAUZERE_tree_kernel}, and planar graphs~\cite{SUN2022127404}, to name a few.  

We now state the problems that will be considered in this paper.
\begin{problem}[$k$-Graphlet Enumeration]
\label{prob1}
    Given an undirected graph $G = (V, E)$, list all the connected subgraphs of $G$ induced by a subset $C \subseteq V $ of $k$ vertices (see~\Cref{fig:subgraph_example:sub:graphlets}).
\end{problem}
\begin{problem}[Edge $k$-Graphlet Enumeration]
\label{prob2}
    Given an undirected graph $G = (V, E)$, list all the connected subgraphs of $G$ induced by a subset $S \subseteq E$ of $k$ edges (see~\Cref{fig:subgraph_example:sub:edge-graphlets}).
\end{problem}


As these problems have an exponential amount of solutions, the \emph{input time complexity} of any algorithm solving these problems is exponential in the worst case; however, when tackling enumeration tasks it is common practice to give an \emph{output-sensitive} analysis instead, where the time complexity as a function of the output size. Specifically, let $N$ be the total number of solutions found. If the overall cost is $O(\mathrm{poly}(|V|) + N \cdot f(k))$ time for some function $f()$, we say that the \emph{amortized} cost per solution,  \changed{i.e., the amortized time complexity} is $O(f(k))$\footnote{$O(\mathrm{poly}(|V|))$ does not \changed{contribute to} the amortized time since it is independent of $N$.}. It is output-sensitive as this cost grows linearly with $N$, the number of solutions. 

Our main result is to provide new enumeration algorithms requiring $O(k^2)$ amortized time per $k$-graphlet in~\Cref{prob1}, and $O(k)$ amortized time per edge $k$-graphlet in~\Cref{prob2}. When compared to the state of the art discussed in~\Cref{sec:related}, \textit{(i)}~we provide a faster \changed{amortized cost}, and \textit{(ii)}~we obtain the first parameterized cost per solution that has no dependency on the graph size $|V|$ or $|E|$, or its maximum degree $\Delta$: typically, order $k$ is a small value whereas $|V|$, $|E|$, and \changed{$\Delta$} can be very large, so our algorithms scale well as their \changed{amortized time complexity} depends on $k$ alone.
\changed{Based on our results, it is also possible to derive \emph{delay} bounds, i.e., to bound the time that interleaves between the output of two consecutive solutions. 
For instance, using simple arguments, we obtain a $O(k^3 + k\Delta)$ delay for the $k$-graphlet enumeration task, but we also how that with more refined techniques we can achieve a $O(k^2)$ delay.}

On the technical side, our contribution is the following. We note that enumeration algorithms based on the so-called binary partition approach often exhibit simplicity in both correctness and algorithmic behavior: binary partition involves determining how to expand each partial solution by either including an additional element $e$ or explicitly excluding $e$. For instance, algorithms for enumerating paths, cycles, spanning trees, induced matchings, induced subtrees, induced subgraphs with a girth constraint, dominating sets, induced bipartite graphs, and $k$-degenerate induced subgraphs, all make use of binary partition (the reader can refer to~\cite{Wasa16}). However, performing a refined analysis of these algorithms presents a significant challenge. In our proposed algorithms for~\Cref{prob2}, we leverage the \emph{push-out amortization} technique introduced in~\cite{DBLP:conf/wads/Uno15}. The crux of this analysis lies in speeding up the lower part of the enumeration tree, where algorithms based on binary partition generate recursive calls, forming a tree structure. In many instances, a recursive call within the lower part of this tree structure deals with an input graph featuring a substantial number of vertices/edges either considered or forbidden for use. In such cases, the possibility arises to eliminate these specified vertices/edges, thereby reducing the size of the input graph. This operation, in turn, accelerates the computation time of the lower part of the enumeration tree.
As the size of an enumeration tree is definitively influenced by the size of its lower section, where the number of descendants increases exponentially, hastening the lower part of an enumeration tree results in exponential improvements in overall efficiency.
A similar idea is called \emph{kernelization}~\cite{Niedermeier:IPL:2000, Damaschke:TCS:2006} in the field of parameterized algorithms, reducing the size of the input graph.

\subsection{Related Work} 
\label{sec:related}

For $k$-graphlets \cite{10.1093/bioinformatics/bth436}, there are numerous studies on techniques for enumerating all of them, those with exactly $k$ vertices and those with \emph{at most} $k$ vertices\footnote{It is also possible to enumerate all vertex-induced connected subgraphs, regardless of the value of $k$: for this task there already exists a constant amortized time algorithm~\cite{DBLP:conf/wads/Uno15}.}, each of these variants being suitable for different kinds of applications.
Historically, one of the first algorithms able to enumerate all $k$-graphlets in a generic graph is called ESU \cite{wernicke_esu, fanmod}, and works with a classical binary partition approach.
Other tools for this task include Escape \cite{escape_algorithm}, FaSE \cite{fase_algorithm}, Kavosh \cite{Kashani2009_kavosh}, and Jesse~\cite{10.1093/bioinformatics/jesse} among others, for which we refer the reader to the survey in \cite{graphlet_survey_2021}. 
An interesting case is the PGD algorithm \cite{PGD_algorithm}, which is able to compute the exact number of $k$-graphlets without having to explicitly enumerate all of them, by some clever combinatorial arguments. 
However, due to the nature of its combinatorial reasoning, PGD is limited to values of $k \leq 4$, because even for $k = 5$ the cases to be considered are too many to have both a reasonable running time and a small number of theoretical arguments to be exploited.

A recent work by Komusuewicz and Sommer \cite{Komusuewicz:Enumerating:2020}, later improved in practice by Conte et al. \cite{practical_graphlets}, sets the state of the art to $O(k^2\Delta)$ delay time per $k$-graphlet, where $\Delta$ is the maximum degree of the input graph $G$, and the delay is the worst case time to output the next solution, and cannot be smaller than the \changed{amortized time}. For the relaxed case of graphlets with \emph{at most} $k$ vertices, the delay becomes $O(k + \Delta)$ time. Wang et al.~\cite{WANG2024106425} have recently proposed three algorithms with delay, respectively, of $\order{k \cdot \min\set{|V|-k, k\Delta}(k\log \Delta +\log |V|)}$, $\order{k \cdot \min\set{|V|
-k, k\Delta}\cdot |V|}$, and $\order{k^2\cdot \min\set{|V|-k, k\Delta}\cdot\min\set{k, \Delta}}$: these bounds can be lower than the state of the art $O(k^2\Delta)$ of~\cite{Komusuewicz:Enumerating:2020}, under some specific assumptions on the value of $k$. 

Compared to the state of the art, our algorithms \changed{presented} in this paper belong to the class of enumeration algorithms \changed{ ``for a given $k$'' }, \changed{with} no limitations on the value of $k$.
Our algorithms improve the current state of the art from roughly $O(k^2 \Delta)$ \changed{amortized time} to $O(k^2)$ per $k$-graphlet or $O(k)$ per edge $k$-graphlet, thus eliminating a troublesome theoretical dependency on $\Delta$. For the sake of completeness, there also exist non-combinatorial algorithms that estimate $k$-graphlet counts in a graph using deep learning techniques \cite{liu2021learning}, but it is difficult to analyze their \changed{amortized time complexity}.

\section{Preliminaries}
In this section we give the necessary notation and terminology about graphs.
In addition, we give an overview of a technique called \emph{binary partition} to design efficient enumeration algorithms.
In \Cref{subsec:bp}, we also show a technique for analyzing the amortized time complexity of enumeration algorithms called \emph{push-out amortization}~\cite{DBLP:conf/wads/Uno15}.

\subsection{Notation}
Let $G = (V, E)$ be an undirected graph. Unless stated otherwise, graphs have no self-loops or parallel edges. 
We may use notations such as $E(G)$ or $V(G)$ to refer to a set of edges or vertices, respectively, for a given graph instance $G$. 
We also use shorthands for $|V|$ and $|E|$ such as $n$ and $m$, respectively.
For an edge $e = \set{u, v}$, we call $u$ and $v$ \emph{endpoints of $e$}.
For a vertex $v$, an edge $e$ is an \emph{incident edge of $v$} if $v$ is an endpoint of $e$.
We denote the set of incident edges of $v$ in $G$ as $\Gamma_G(v)$.
The value $\size{\Gamma_G(v)}$ is called the \emph{degree of $v$} and denote it as $d_G(v)$.
Moreover, $\max_{v \in V}d_G(v)$ is the \emph{maximum degree of $G$} and denote it as $\Delta$.
We drop the subscript when it can be \changed{deduced} to from the context.
A vertex $v$ is a \emph{neighbor of $u$} if $E$ has an edge $\set{u, v}$.
The set of neighbors of $v$ is the \emph{open neighborhood of $v$} and denote it $N(v)$, and
we define the open neighborhood of a set of vertices $U$ as $\bigcup_{u \in U}N(u) \setminus U$.
Similarly, we define the \emph{closed neighborhood of $v$ and $U$} as $N[v] \coloneqq N(v) \cup \set{v}$ and $N[U] \coloneqq \bigcup_{u \in U}N(u) \cup U$. 
\changed{This reasoning can be extended also to a set of edges $I$: in particular, given $I \subseteq E$, we define $\Gamma(I)$ as the set of edges incident to at least one of the endpoints of an edge $e \in I$, i.e. $\Gamma(I) = \inset{e = (u, v) \in E}{\exists e' = (u', v') \in I : u = u' \vee u = v' \vee v = u' \vee v = v' }$.}

A vertex sequence $P \coloneqq (v_1, \ldots, v_\ell)$ is a \emph{path} if for any $1 \le i < \ell$, $G$ has an edge $\set{v_{i}, v_{i+1}}$ and $v_i$ and $v_j$ are distinct for any $1\le i < j \le \ell$.
Vertices $v_1$ and $v_\ell$ \changed{are} \emph{endpoints of $P$} and $P$ is also called a $v_1$-$v_\ell$ path.
The \emph{length} of path $P$ is the length of the sequence minus one\changed{, i.e., the number of edges traversed}.
The \emph{distance} between vertices $u$ and $v$ is the length of a shortest $u$-$v$ path.
A graph $G$ is \emph{connected} if, for any pair of vertices $u$ and $v$, $G$ has a $u$-$v$ path.
\changed{
Moreover, we define $dist(X, e)$, the distance of an edge $e = (u, v)$ from a set of vertices $X$, as the length of a shortest path from either $u$ or $v$, depending on which one is closer, to the closest vertex belonging to $X$.
In this paper we assume that the graph $G$ is connected. If $G$ is not connected, it suffices to consider each connected component of $G$ separately.}
A sequence $C \coloneqq (v_1, \ldots, v_\ell, v_1)$ is a \emph{cycle}
if the sequence $(v_1, \ldots, v_\ell)$ is a path and $G$ has an edge $\set{v_1, v_\ell}$.
For $G$ and a subset of vertices $U \subseteq V$,
a graph $H = (U, F)$ is called a \emph{vertex-induced subgraph} if $F \coloneqq \inset{\set{u, v} \in E}{u,v \in V}$.
Similarly, 
for $G$ and a subset of edges $F \subseteq E$,
a graph $H = (U, F)$ is called an \emph{edge-induced subgraph} or simple \emph{subgraph} if $U \coloneqq \bigcup_{e \in F}e$.
We denote a vertex- or edge-induced subgraph induced by $X$ as $G[X]$. 
\changed{Moreover, we denote a vertex- or edge-induced subgraph induced by $V \setminus X$ or $E \setminus X$ as $G - X$.
For a vertex $v$ (an edge $e$), we denote $G - \set{v}$ ($G - \set{e}$) as $G - v$ ($G - e$). We also define the \textit{contraction} of an edge $e = \{u,v\}$ as the process of removing one of its extremes, e.g., $v$, and inserting a new edge from $u$ to all the vertices in $\Gamma(v)$; we remark that contraction can sometimes create multi-edges.} 

We recall that in this paper, we call a connected vertex- or edge-induced subgraph with exactly $k$ vertices or edges a \emph{$k$-graphlet} or an \emph{edge $k$-graphlet}.
An edge subtree $T$ is a \emph{subtree of $G$} if $T$ has no cycles. If $T$ has exactly $k$ edges, we call $T$ a \emph{$k$-subtree} (see \Cref{fig:subgraph_example} for a visualization of these three concepts). 

\subsection{Basic techniques for efficient enumeration}\label{subsec:bp}

\changed{Let us introduce} in this section the \changed{main algorithmic techniques} used in the paper.

\emph{Binary Partition} (also \emph{branching} or \emph{flashlight} approach), is a basic technique that recursively splits the space of solutions $\mathcal S$ into two sets $\mathcal S'$ and $\mathcal S''$ such that $\mathcal S = \mathcal S' \cup \mathcal S''$.
Consider for example the enumeration of all $k$-graphlets in a graph $G = (V, E)$.
We first pick a vertex $v \in V$ and partition $\mathcal S$, the space of all $k$-graphlets of $G$, into two sets: one containing all solutions that include $v$, the other containing all solutions that do not include $v$.
If the aforementioned sets contain more than one solution, we recursively proceed until we can print a unique $k$-graphlet.
This process defines an implicit tree structure $\mathcal T = (\mathcal V, \mathcal E)$ that we call the \emph{enumeration tree} or recursion tree.
We will refer to a $X \in \mathcal V$ as a tree \emph{node} or an \emph{iteration}, and we call a node $Y \in \mathcal V$ a \emph{child} of $X$ if and only if there exists an arc $(X, Y) \in \mathcal E$. 
The set of all children of $X$ is called $ch(X)$.
\changed{The parent-child relation of $\mathcal T$ can be defined as follows. Let $X$ be a node of $\mathcal T$, corresponding to an instance of the problem on the graph $G_X = (V_X, E_X)$. 
The children nodes of $X$ either consider $v \in V$ as part of the solution, or remove it from the graph. 
Thus, the left child $X_l$ of $X$ obtains the same input graph $G_X$, but includes $v$ in all subsequent solutions (e.g. by keeping it in a special set $S$), while the right child obtains the input graph $G_{X_r} = (V_X \setminus \{v\}, \inset{(u, w) \in E_X}{u,w \in V_X \setminus \{v\}}$.
}
A key factor when designing an efficient enumeration algorithm based on this technique, is to ensure that each internal node in $\mathcal T$ has at least two children.
Suppose that any internal node $X \in \mathcal T$ can generate all its children in $T(X)$ time, and that each leaf node produces exactly one solution. 
When every internal node has exactly two children, the number of solutions, i.e. of leaves, linearly bounds the number of internal nodes, allowing us to conclude that the \emph{amortized time complexity} of this enumeration algorithm is $\order{\max_{X \in \mathcal T} T(X)}$.
Notice that if there exists an internal node with exactly one child in $\mathcal T$, the amortized time complexity of the algorithm worsens to $\order{d \cdot \max_{X \in \mathcal T} T(X)}$, where $d$ is the depth of $\mathcal T$.

This technique has been widely used in the literature for the enumeration of structures like $s$-$t$ paths, cycles, minimal $s$-$t$ cuts, spanning trees, induced matchings, induced bipartite subgraphs, induced $k$-degenerate subgraphs, independent sets, and more~\cite{Read:Tarjan:Networks:1975,Birmele:Optimal:2013,Tsukiyama:Shirakawa:JACM:1980,Wasa:COCOON:2018,Kazuhiro:IEICE:2018,10.1007/978-3-030-04651-4_3,DBLP:journals/corr/abs-1906-09680}.
In this paper, we will use this strategy for designing the algorithm of \Cref{sec:graphlet}.

We now introduce another technique that can be used for the time complexity analysis of binary partition algorithms called \emph{push-out amortization}~\cite{DBLP:conf/wads/Uno15}.
The idea of this kind of analysis is to speed up the computation of the nodes in the bottom level of the enumeration tree. 
In fact, as each internal node has at least two children, the number of nodes in the bottom part is exponentially larger than that of nodes in top part.
Therefore, by speeding up the lowest nodes, we expect to exponentially improve the time complexity of the whole algorithm. 
With this idea in mind, one has to design an algorithm respecting the following conditions: (1) the overall time complexity is dominated by the complexity of leaf nodes, and {(2)\label{po-condition_2}} leaves have ``small'' time complexity.
For reference, the independent set enumeration problem can also be solved efficiently by push-out amortization, as well as other enumeration problems \cite{DBLP:journals/corr/abs-1906-09680, DBLP:conf/wads/Uno15}.

In order to prove that the total time of the bottom levels of an enumeration tree dominate the total running time of an algorithm, it is important to ensure that the total time of children nodes is somewhat \emph{larger} than that of the parent node.
However, at the same time we must also ensure that the running time of leaf nodes is kept small, and satisfying both of these conditions is not always easy.
This is because usually when designing an algorithm in a way that the time complexity of a child node is large, it is difficult to ensure condition (2). 
On the other hand, when we design the leaves with a small time complexity, it becomes difficult to ensure condition (1).
Uno \cite{DBLP:conf/wads/Uno15} established 
a sufficient condition to solve the above situation called the \emph{push-out (PO) condition}.

\begin{theorem}[The PO condition~\cite{DBLP:conf/wads/Uno15}]\label{thm:push-out}
    If, for every internal node $X$ of $\mathcal T$, it holds that
    \[
        \Bar{T}(X)=\sum_{Y \in ch(X)}T(Y) \quad \ge \quad \alpha T(X) - \beta(|ch(X)| + 1)T^*,
    \]
    where $\alpha > 1$ and $\beta \ge 0$ are constants and $T^*$ is the worst case running time of processing any leaf node of $\mathcal T$,
    then, the amortized time of each node $X$ in $\mathcal T$ is $O(T^*)$.
\end{theorem}

Suppose, for example, that for every leaf node of $\mathcal T$, the algorithm outputs a distinct solution (we assume, for the sake of convenience, that one solution corresponds to exactly one leaf and vice versa). 
If every internal node of $\mathcal T$ has at least two children and $T^*$ is upper bounded by a constant, then the algorithm runs in amortized constant time per solution.
The algorithm presented in \Cref{sec:es} is designed according to the push-out amortization strategy.

\changed{Finally, a note on data structures is in order: we are describing recursive algorithms that make changes to a graph as they proceed. Since creating copies of the graph for children recursive calls would dominate the complexity, we keep a single dynamic copy of the graph and edit it as necessary, undoing the changes as we backtrack.}

\changed{We assume using a graph data structure that allows the key operation of
\begin{itemize}
    \item Iterating over the neighbors of a node $v$ in $O(d_G(v))$ time.
    \item Deleting or inserting an edge in $O(1)$ time.
    \item Deleting or inserting a vertex $v$ in $O(d_G(v))$ time.
    \item Contracting an edge $e = \{u,v\}$ in either $O(d_G(u))$ or $O(d_G(v))$ time.
\end{itemize}
In particular observe how the first two items are the essential ones, as we can delete a vertex $v$ in $O(d_G(v))$ time by deleting its incident edges one by one, and we can contract $e = \{u,v\}$ in $O(d_G(u))$ time by removing every edge $\{u,x\}$ incident to $u$ and replacing it with the edge $\{v,x\}$.
This type of structure is typically considered folklore knowledge, but more details on how to maintain it can be found, e.g., in \cite{wasa2012constant}.
}

\subsection{\changed{Obtaining Delay Bounds}}
\changed{As pointed out earlier, it is possible to transform these amortized time bounds into delay. Let us briefly observe how.
In particular, we can use the \emph{Output Queue Method} \cite{uno2003two}.
It consists in creating a bounded queue $Q$ of length $\left\lceil 2T^*/\Bar{T}\right\rceil+1$, where $T^*$ is the maximum delay, i.e. the cost of all the ancestors of a leaf iteration in the recursion algorithm that outputs a solution, and $\Bar{T}$ the average time between two outputs (i.e. the amortized cost).
During the execution of our algorithm, instead of outputting a solution directly, we place the solution into the end of $Q$. 
If $Q$ is full, then we pop the front and output that solution.
If we reach the end of the algorithm and $Q$ is not empty, output all the remaining solutions.
By using this technique we can regularize the delay and make sure to output every a solution every $\Bar{T}$ time. In particular, we can obtain $O(k^2)$ and $O(k)$ delay respectively for $k$-graphlet enumeration and edge $k$-graphlet enumeration. 
We can then proceed to give amortized bounds on our algorithms, knowing that they can always be translated into delays with the Output Queue Method.
For more details on this technique, we refer the reader to \cite{uno2003two}.
In what follows, we will not explicitly adopt this technique in the pseudocode in order not to further complicate the analysis of our algorithms, but we note that the regularized delay is achievable through small modifications to each algorithm.
}

\section{Enumeration of \texorpdfstring{$k$}{k}-graphlets}\label{sec:graphlet}

In this section we describe an algorithm for efficient $k$-graphlet enumeration in amortized $O(k^2)$-time, based on binary partition. 
We start by discussing an algorithm for the enumeration of $k$-graphlets containing a given vertex $r$, then we expand the strategy to obtain all the $k$-graphlets of $G$.

\subsection{Enumeration of all \texorpdfstring{$k$}{k}-graphlets with a specific vertex}\label{subsec:graphlets-simple-binary-partition}

Let $r \in V$ be an arbitrary vertex of $G = (V, E)$, $\In$ and $\Ex$ be two disjoint vertex sets, and $\mathcal S(G, \In, \Ex, k)$ be the set of all $k$-graphlets of $G$ that contain all and only the vertices in the set $\In$, and do not contain any vertex in $\Ex$.
For brevity, when $G$ is clear from the context, we denote this set as $\mathcal S(\In, \Ex, k)$.
Since we are interested in $k$-graphlets containing the specified vertex $r$, in the rest of this section we always assume that $r \in \In$. 
The basic strategy consists of partitioning the set $\mathcal S(\In, \Ex, k)$ into two sets $\mathcal S(\In \cup \set{v}, \Ex, k)$ and $\mathcal S(\In, \Ex \cup \set{v}, k)$, given a vertex $v \in N(\In) \setminus \Ex$.
Then, we check if these sets contain at least one solution by choosing an arbitrary vertex from $\In$, running a graph traversal from it, excluding vertices in $\Ex$, and checking if we reach $k - \size{\In}$ vertices to be added to $\In$.
This strategy ensures that $G[\In]$ is connected, as it is connected at first (initially $\In = \{r\}$) and we add neighbors of $\In$ at each step. 

Therefore, in what follows, we assume that $G[\In]$ is connected, and we also omit the set $\Ex$ from the notation using this simplifying observation.

\begin{observation}
    $\mathcal S(G, \In, \Ex, k) = \mathcal S(C_\In, \In, \emptyset, k)$, 
    where $C_\In$ is the connected component in $G - \Ex$ that contains all the vertices in $\In$.
\end{observation}
\changed{In what follows, unless otherwise specified, we assume that the graph is stored using adjacency lists, providing $O(d(v))$-time deletion/restore of any vertex $v \in V$. 
In the algorithms we represent a graphlet implicitly, by keeping the set of its vertices.
Any induced subgraph $G[S]$ is not explicitly represented, and we can traverse it by starting from a vertex in the corresponding set $S$.}

\subsubsection{Amortized linear-time enumeration}

First, we give an amortized $\order{m}$-time enumeration algorithm that will be used as a subroutine in the subsequent $\order{k^2}$ amortized time algorithm. 

It can be easily shown that the \changed{simple binary partition algorithm discussed in Section \ref{subsec:graphlets-simple-binary-partition}} runs in amortized $\order{nm}$ time, since each recursive call takes $\order{m}$ time and the depth of $\mathcal T$ is at most $n$.
In this approach, the number of tree nodes in $\mathcal T$ is $\order{nN}$, where $N$ is the number of leaves, that is, the number of $k$-graphlets.
Therefore, the amortized time complexity is $\order{nm}$.
To improve this time complexity, we first reduce the number of nodes in $\mathcal T$ without worsening the time complexity of each node and, in particular, we ensure that each node produces exactly two children.
To this end, we detect a vertex $v$ such that $\mathcal S(\In, \Ex \cup \set{v}, k)$ has no solution.

A vertex $v$ is \emph{mandatory} for $G$, $\In$ and $k$ if the connected component in $G - v$ that contains the vertices in $\In$ has at most $k-1$ vertices.
Notice that a vertex $v$ is mandatory if and only if $\mathcal S(G, \In, k) = \mathcal S(G, \In \cup \set{v}, k)$.
In addition, a mandatory vertex is an articulation point if the connected component in $G$ that contains $\In$ has at least $k+1$ vertices.

In our partitioning strategy, the basic operation is to determine if $\mathcal S(G, \In, k)$ is non-empty, and if a vertex $v$ is mandatory or not.
We show that both operations can be done in $O(\min\set{m, k^2, k\Delta})$ time.

\begin{lemma}\label{lemma:removable_check_time}
    Let $G$ be a graph, $k$ be an integer, and \changed{$\In \neq \emptyset$} a set of vertices such that $G[\In]$ is connected.
    We can determine if $\mathcal S(G, \In, k)$ is empty or not in $O(\min\set{m, k^2, k\Delta})$ time.    
    Moreover, for a vertex $v$, we can check whether $v$ is mandatory or not with the same time complexity.
\end{lemma}
\begin{proof}
    We show that it can be done with a standard graph search technique, the breadth-first search.
    If the connected component in $G$ that contains $\In$ has at least $k$ vertices, 
    $\mathcal S(G, \In, k)$ is non-empty, and this can be checked in $\order{m}$ time.
    To improve this time complexity, we terminate the breadth-first search as soon as we find $k$ vertices.
    Let $U$ and $F$ be two sets of vertices and edges respectively, explored by the breadth-first search.
    Since the breadth-first search terminates when we find $k$ vertices, each edge $e \in F$ connects vertices in $U$.
    Therefore, $\size{F}$ is bounded by $\min\set{k^2, k\Delta}$ and the breadth-first search terminates in $\order{\min\set{m, k^2, k\Delta}}$ time.
    Moreover, even if we ignore $v$, the number of edges in $F$ increases at most by $k$, thus
    we can check whether $v$ is mandatory or not with the same time complexity.
\end{proof}

We are now ready to describe a linear amortized time algorithm for $k$-graphlet enumeration, whose pseudocode is shown in \Cref{alg:linear_enum_graphlet}.
We first prove its correctness, then we discuss its computational complexity.

\begin{algorithm}[tb]
  \DontPrintSemicolon
  \KwIn{A graph $G$, a set of vertices $\In$, an integer $k$}
  \KwOut{All $k$-graphlets that contain $\In$ in $G$}
  \SetKwProg{myproc}{Function}{}{}
  \SetKwFunction{enum}{LIN-ENUM-V}

  \myproc{\enum{$G$, $\In$, $k$}}{%
      \lIf{$\size{\In} = k$}{\textbf{output} $\In$ \textbf{and} \Return}
      \While{$N(\In)$ has a mandatory vertex $u$}{\label{alg:while:lin}
         $\In \gets \In \cup \{u\}$\;
         \lIf{$\size{\In} = k$}{%
           Output $\In$ and \Return
        }%
      }
      Let $v$ be a vertex in $N(\In)$ \tcp*[r]{$v$ is non-mandatory} 
      \enum{$G, \In \cup \set{v}, \emptyset, k$}\;
      \enum{$G - v, \In, k$}\;
  }%
  \caption{An amortized $O(m)$-time algorithm for $k$-graphlet enumeration}
  \label{alg:linear_enum_graphlet}
\end{algorithm}

\begin{theorem}\label{th:correctness_linear_enum}
    \Cref{alg:linear_enum_graphlet} enumerates all $k$-graphlets in $G$ \changed{that contain $\In \neq \emptyset$}.
\end{theorem}
\begin{proof}
    Let $X$ be a node in $\mathcal T$, $G_X$ be the input graph of the call $X$, $\In_X$ be a set of vertices such that $G_X[\In_X]$ is connected.
    We show that each node in $\mathcal T$ partitions 
    $\mathcal S(G_X, \In_X, k)$ into two non-empty sets
    $\mathcal S(G_X, \In_X \cup \set{v}, k)$ and 
    $\mathcal S(G_X - v, \In, k)$.
    By definition, adding mandatory vertices in $\In_X$ does not change the set of $k$-graphlets.
    Moreover, since we add mandatory vertices that are adjacent to a vertex in $\In_X$, we do not break the connectivity of $G_X[\In_X]$. 
    Let $S$ be a $k$-graphlet in $\mathcal S(G_X, \In_X, k)$: if $S$ contains $v$, $S \in \mathcal S(G_X, \In_X \cup \set{v}, k)$,
    otherwise $S \in \mathcal S(G_X - v, \In_X, k)$.
    From the definition of $\mathcal S(G_X, \In_X, k)$, 
    $\mathcal S(G_X, \In_X \cup \set{v}, k)$ and 
    $\mathcal S(G_X - v, \In_X,  k)$ are disjoint.
    Therefore, in each node $X$ of $\mathcal T$, we can partition the set of $k$-graphlets and \Cref{alg:linear_enum_graphlet} output all $k$-graphlets containing $\In$.
\end{proof}

Next, we analyze the time complexity of \Cref{alg:linear_enum_graphlet}.

\begin{theorem}\label{th:linear_graphlet_enum_delay}
    \Cref{alg:linear_enum_graphlet} runs in $O(n + N\cdot(m + 1))$ time and $\order{n + m}$ space, where $N$ is the number of $k$-graphlets in $G$.
\end{theorem}
\begin{proof}
    We first consider the time complexity of each node $X$, i.e., of a single recursive call.
    First, we enumerate all articulation points in 
    $O(n + m)$ time using Tarjan's articulation point enumeration algorithm~\cite{DBLP:journals/ipl/Tarjan74, tarjan_articulation_points}.
    Using the block-cut tree obtained from Tarjan's algorithm, we can obtain all mandatory vertices in $\order{n + m}$ time; thus, the while loop on line~\ref{alg:while:lin} takes $\order{n + m}$ time.
    Since the other operations can be done trivially in $\order{n + m}$ time, the time complexity of $X$ is $\order{n + m}$.
    Now we consider the total time complexity of \Cref{alg:linear_enum_graphlet} by analyzing the relation between the number of solutions and the number of nodes of $\mathcal T$.
    Since each leaf of $\mathcal T$ outputs one solution and each internal node of $\mathcal T$ has at least two children,
    the number of leaves of $\mathcal T$ is bounded by $\order{N}$.
    Therefore, the total time complexity of this algorithm is $\order{n + N(m + 1)}$.
    Regarding the space complexity,
    in each node $X$ of $\mathcal T$, we store only the difference of $\In_X$ and $G_X$, achieving a
    total space complexity of $\order{n + m}$.
\end{proof}
\subsubsection{Amortized \texorpdfstring{$\order{k^2}$}{O(k2)}-time enumeration}
We can now introduce our amortized $\order{k^2}$-time $k$-graphlet enumeration algorithm.
The basic strategy of this algorithm is the same as \Cref{alg:linear_enum_graphlet}, fine-tuned in order to achieve amortized $\order{k^2}$-time.
The enumeration of mandatory vertices demands $\order{m}$ time if done in a naive way, using a block-cut tree and a linear-time articulation point enumeration algorithm~\cite{tarjan_articulation_points}: this may become the bottleneck of the enumeration process.
To overcome this difficulty, we use \Cref{lemma:removable_check_time}.
Using this Lemma, we can efficiently determine whether a vertex $v$ is mandatory or not.
If we can find a non-mandatory vertex, then we can generate two children instances in the corresponding recursion tree $\mathcal T$.
We now show that $N(\In)$ has at most \changed{one mandatory vertex} when $G$ is sufficiently large.

\begin{algorithm}[tb]
  \DontPrintSemicolon
  \KwIn{A graph $G$, an integer $k$, a \changed{connected} set of vertices $\In$}
  \KwOut{All $k$-graphlets of $G$ that contain $\In$}
  \SetKwProg{myproc}{Function}{}{}
  \SetKwFunction{enum}{ENUM-V}
  \SetKwFunction{removable}{REMOVABLE}
  \SetKwFunction{linearenum}{LIN-ENUM-V}
  \SetKwFunction{fruitful}{FRUITFUL}

  \myproc{\enum{$G$,$\In$,$k$}}{%
    \If{$\size{\In} = k - 1$}{
        \lFor{$u \in N(\In)$}{
            \textbf{output} $\In \cup\set{u}$\;
            \Return
        }
    }
    \While(\tcp*[f]{Follow the chain}){$\size{N(\In)} = 1$}{\label{alg:ksquare_graphlet:line:while}
      Let $v$ be the vertex in $N(\In)$\;
      $\In \gets \In \cup \{v\}$\;
      \If{$\size{\In} = k - 1$}{
        \lFor{$u \in N(\In)$}{
            \textbf{output} $\In \cup\set{u}$\;
            \Return
          }
      }
    }
    Let $x,y$ be two distinct vertices in $N(\In)$\;
    \uIf{either $x$ or $y$ is non-mandatory for $G, \In, k$}{
        Let $z$ be a non-mandatory vertex in $\set{x, y}$\;
      \enum{$G, \In \cup \{z\}, k$}\;\label{alg:ksqare_graphlet:line:take_x}
      \enum{$G - z, \In, k$}\;\label{alg:ksqare_graphlet:line:remove_x}
    }
    \Else(\tcp*[f]{$|V| < 2k$ by \Cref{lemma:g_2k_nodes}}){
      \linearenum{$G, \In, k$}\;
    }
  }
  \caption{An amortized $O(k^2)$-time $k$-graphlet enumeration algorithm.}\label{alg:ksquare_graphlet}
\end{algorithm}


\begin{lemma}\label{lemma:g_2k_nodes}
    Let \changed{$U \neq \emptyset$} be a set of vertices \changed{of a graph $G$} such that $G[U]$ is connected, \changed{$\size{U} \le k-2$, and $G$ has at least $k$ vertices.  Let $x$ and $y$ be distinct vertices in $N(U)$}.
    If $x$ and $y$ are both mandatory for $G$, $U$, $k$, then $n < 2k$.
\end{lemma}
\begin{proof}
    Since $x$ and $y$ are mandatory, both $x$ and $y$ must be articulation points in $G$, \changed{and they are in the connected component of the graph containing $U$ (because they are neighbors of vertices of $U$).}
    Let $T_C$ be a biconnected-component tree of $G$: \changed{this component must contain at least $k+1$ vertices, otherwise it would not contain any solution and we would skip it (see Algorithm~\ref{alg:graphlet_instances}).}
    We denote the connected component in $G - x$ that contains the vertices in $U$ as $T_C(x, U, 1)$  and
    the other connected components in $G-x$ as $T_C(x, U, 0)$.
    Notice that $x \in T_C(y, U, 1)$ and $y \in T_c(x, U, 1)$ by construction.
    Since both $x$ and $y$ are mandatory, 
    both $T_C(x, U, 1)$ and $T_C(y, U, 1)$ have at most $k - 1$ vertices.
    Moreover, any vertex in $G$ is contained either in $T_C(x, U, 1)$ or $T_C(y, U, 1)$, \changed{because $x$ and $y$ are both mandatory. As any shortest path from $x$ to $y$ can only go through vertices of $U$, because of the size of $G$, we obtain that the corresponding $T_C(x, U, 1)$ and $T_C(y, U, 1)$ are disjoint, otherwise $x$ and $y$ would not be mandatory.}
    Therefore, the number of vertices of $G$ in this case is at most 
    $\size{V(T_C(x, U, 1))} + \size{V(T_C(y, U, 1))} \le 2k - 2$.
\end{proof}


We can now prove the amortized time complexity of~\Cref{alg:ksquare_graphlet}.
We omit the proof of correctness for  \Cref{alg:ksquare_graphlet} as it follows the same arguments used in \Cref{th:correctness_linear_enum} \changed{, and from the binary partition structure of the algorithm (see also \cite{Komusuewicz:Enumerating:2020}).}

\begin{theorem}\label{th:ksquare-graphlet-complexity}
    \changed{\enum{$G$,$\In=\{v_i\}$,$k$} (\Cref{alg:ksquare_graphlet}) runs in $\order{N_i\cdot k^2}$ time and $\order{n + m}$ space, where $N_i$ is the number of $k$-graphlets that contain $\In=\{v_i\}$.}
\end{theorem}
\begin{proof}
    We analyze the time complexity of each \changed{tree} node $X \in \mathcal T$, and we denote by $G_X$ the input graph of the call $X$, with the corresponding set $\In_X$.
    Notice that $k$ is the same for all instances in $\mathcal T$\changed{, and that $X$ is not a leaf node of the recursion tree}.
    We first analyze the time complexity when $\size{\In_X} = k-1$ (there cannot be a node satisfying $\size{\In_X} = k$ \changed{as the base case of the algorithm is for $\size{\In_X} = k-1$})\footnote{We could also use $\size{\In} = k$ as the base case, and the overall time complexity would not change. However, the time complexity analysis would become more difficult to follow, and we would introduce one additional level of recursion which we think it is useful, and meaningful, to avoid.}.
    In this case, we can enumerate all $k$-graphlets in $\order{k^2 + k \cdot N_X}$ time, where $N_X = \size{\mathcal S(G_X, \In_X, k)}$\changed{, and the node $X$ produces only leaf children.}
    For each $u \in N_{G_X}(\In)$, $\In \cup \set{u}$ is a $k$-graphlet.
    Therefore, $\sum_{v \in S} d_{G_X}(v)$ is upper-bounded by $k^2 + k\cdot N_X$.

    Now consider the case $\size{\In_X} < k - 1$.
    When $N_{G_X}(\In_X)$ has only one vertex, we can find \changed{it} in $\order{k^2}$ time \changed{by looking at the neighbors of the set and stopping as soon as we find one not in the set}, and it is easy to compute all such vertices in the same time complexity \changed{within the loop of line \ref{alg:ksquare_graphlet:line:while} of \Cref{alg:ksquare_graphlet}}.
    Thus, in what follows we assume that $N_{G_X}(\In_X)$ contains at least two vertices\changed{, and that the tree node $X$ produces non-leaf children}.
    Let $x$ and $y$ be two distinct vertices in $N_{G_X}(\In_X)$: using \Cref{lemma:removable_check_time}, we can check if $x$ and $y$ are mandatory in $\order{k^2}$ time.
    If both $x$ and $y$ are mandatory, the time complexity of $X$ is $\order{k^2}$ since the number of vertices in $G_X$ is $\order{k}$.
    
    Suppose now that one of $x$ or $y$ is non-mandatory and call $z$ such vertex. 
    To compute $G_X - z$, we need $\order{d_{G_X}(z)}$ time using standard data structures.
    Therefore, if $d_{G_X}(z) < 3k$, the time complexity of $X$ is $\order{k^2}$.
    In the case that $d_{G_X}(z) \ge 3k$, we consider the following amortized analysis:
    let $Y$ be a node generated by $G_X$ and $\In_X \cup \set{z}$.
    We call \emph{left child} a child of a recursive node generated by adding one vertex to $\In$ (i.e., in line~\ref{alg:ksqare_graphlet:line:take_x}), whereas we call \emph{right child} a child of a recursive node generated by removing a vertex from $G$ (i.e., in line~\ref{alg:ksqare_graphlet:line:remove_x}).
    For a node $Z \in \mathcal T$, we call a path $\mathcal P$ from $Z$ to a leaf in $\mathcal T$ obtained by only right children the \emph{right path of $Z$} (see \Cref{fig:graphlet-amortization} for an example).
    We show that the right path of $Y$ has at least $d_{G_X}(z) - 2k$ nodes. 
    If it holds, by charging the $O(d_{G_X}(z))$ time cost divided equally on each node of the path, it amounts to $O(1)$ per each node, so the amortized time complexity of $X$ becomes $\order{k^2}$ whereas the complexity of all other nodes is not changed.
    Since $d_{G_X}(z) \ge 3k$ and $\size{\In_X} < k-1$,
    $N_{G_X}(\In_X \cup \set{z})$ has at least $d_{G_X}(z) - \size{\In_X} \ge 2k+2$ vertices.
    Notice that, for a node $Z$, if $N_{G_Z}(\In_Z)$ has $k - \size{\In_Z}$ vertices, there is no mandatory vertex.

    \changed{Therefore, by picking and removing $k$ vertices in $N_{G_X}(\In \cup \set{z})$}, we cannot create new mandatory vertices; this means that the right path of $Y$ has at least $d_{G_X}(z) - 2k \ge k$ nodes.
    As each node is charged at most once, the time complexity of each node is $\order{k^2}$.
    We then conclude that \Cref{alg:ksquare_graphlet} has total time complexity of $\order{N_i \cdot k^2}$.

    We next consider the space complexity.
    Let $X$ be a node and $Y$ be a child of $X$.
    In each recursive call, it is sufficient to store the difference between $\In_X$ and $\In_Y$ and $G_X$ and $G_Y$.
    It is clear that the sum of the difference between $\In_X$ and $\In_Y$ is bounded by $k$ and the difference between $G_X$ and $G_Y$ is $\order{d(z)}$, where $z$ is a removed vertex.
    Therefore, the total space complexity is bounded by $\order{n + m}$.
\end{proof}

\begin{figure}[htb]
    \centering
    \includegraphics[width=.5\textwidth]{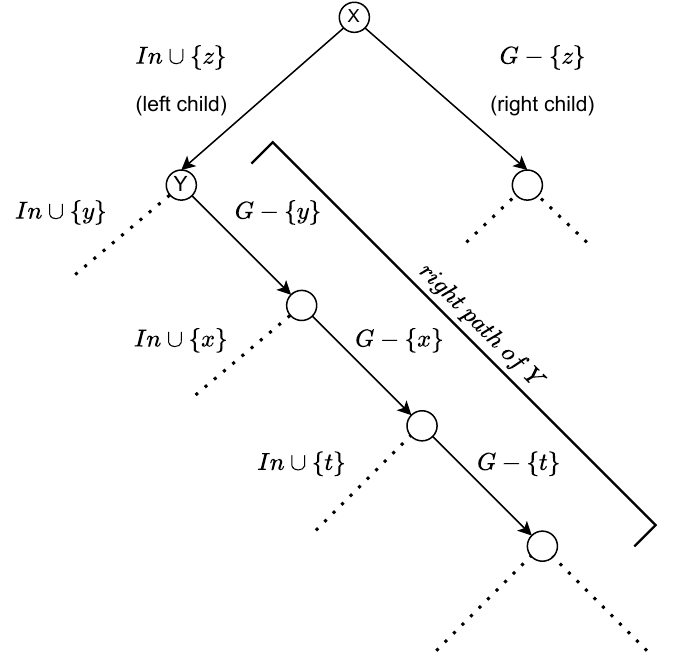}
    \caption{Amortizing the cost of the recursive call $X$ on the right path of $Y$: we charge $O(1)$ to each call on the nodes highlighted by the bracket. These calls are at least $d_{G_X}(z) - 2k$, so we can amortize the $O(d_{G_X}(z))$ cost on them.}
    \label{fig:graphlet-amortization}
\end{figure}

\subsubsection{Bounded-degree graphs}
The focus of the above analysis is to provide bounds that depend only on $k$ and not on the size of the graph $G$.
As a byproduct, it is worth mentioning that the time complexity of this algorithm can be formalized more tightly for bounded-degree graphs, i.e., graphs where the maximum degree $\Delta$ is bounded by a constant.
In the analysis in \Cref{th:ksquare-graphlet-complexity}, 
we consider that the time complexity of finding a mandatory vertex is $\order{k^2}$.
From \Cref{lemma:removable_check_time}, it can be done in $\order{k\cdot \min\set{k, \Delta}}$ time.
Moreover, since removing one vertex can be done in $\order{\Delta}$ time, each recursive call can be done in $\order{k \cdot \min\set{k, \Delta}}$-time without an amortized analysis.
We can thus show a tighter time complexity bound by taking this into account.

\begin{theorem}\label{th:bounded:degree}
    \changed{\Cref{alg:ksquare_graphlet} runs in $O(N_i \cdot k \cdot \min\{k, \Delta\})$ time, where $N_i$ is the number of $k$-graphlets in $G$ that contain $\In \{v_i\}$. Thus, if $\Delta = O(1)$, \Cref{alg:ksquare_graphlet} runs in $O(k \cdot N_i)$ time.}
\end{theorem}
\begin{proof}
    The \changed{bottlenecks} of \Cref{alg:ksquare_graphlet} \changed{are} the deletion of a vertex, that can be trivially done in $\order{\Delta}$ time, \changed{and the identification of mandatory edges, that by \Cref{lemma:removable_check_time} takes $O(\min\set{k, \Delta})$ time.}
    Therefore, each recursive call can be performed in $\order{k\cdot \min\set{k, \Delta}}$ time and the statement holds.
\end{proof}

\subsection{\changed{Generating all \texorpdfstring{$k$}{k}-graphlets using \texorpdfstring{Algorithm~\ref{alg:ksquare_graphlet}}{Algorithm 2}}}
\Cref{alg:ksquare_graphlet} enumerates all $k$-graphlets in $G$ that contain a specific set of vertices $\In$, in \changed{$O(k^2)$ amortized time.}
It remains to see how to generate instances for it, i.e., how to generate \emph{all} $k$-graphlets of $G$ while retaining the same amortized time complexity. 
To do this we rely on \Cref{alg:graphlet_instances}\changed{: Observe how vertices $v_1,\ldots,v_n$ are considered in a \textit{reversed Breadth-First Search (BFS) order}, meaning we can perform a BFS from any arbitrary vertex, write down the vertices in the order they are found in, then reverse this order. The advantage is that every \textit{suffix} $v_i,\ldots,v_n$ is a connected subgraph, so it will always allow finding a $k$-graphlet as long as there are at least $k$ vertices left (which is why we stop at $i = n-k+1$).
We can now} prove the overall complexity for $k$-graphlet enumeration:



\begin{algorithm}[htb]
\DontPrintSemicolon
    \caption{\changed{Generate all $k$-graphlets of a graph $G$ using \Cref{alg:ksquare_graphlet} as subroutine.}}
    \label{alg:graphlet_instances}
    \KwIn{A connected graph $G = (V, E)$, a positive integer $k$}
    \KwOut{All $k$-graphlets.}
    \SetKwFunction{GenIns}{Gen-Ins}{}{}
    \Fn{\GenIns{$G = (V, E), k$}}{
        \changed{Let $v_1, \ldots, v_n$ be a reversed BFS order of $V$} \;
        $G_0 \gets G$\;
        \For(\tcp*[h]{Skip the final $k-1$ vertices}){\changed{$i = 1, 2, \ldots, (n-k+1)$}}{ 
            \changed{\enum{$G_{i-1}$,$\{v_i\}$,$k$}}\;
            $G_i \gets G_{i-1} - \set{v_i}$\;
        }
        
    }
\end{algorithm}


\begin{theorem}\label{th:wrapper_graphlets}
    Let $G = (V, E)$ be a connected graph and $k$ a positive integer. 
    We can enumerate all $k$-graphlets in $G$ in \changed{$\order{n + m + k^2\cdot N}$} total time and $\order{n + m}$ space, where $N$ is the number of solutions.
\end{theorem}
\begin{proof}
    \changed{We show that we can enumerate all $k$-graphlets using \Cref{alg:ksquare_graphlet} as a subroutine.
    Firstly, the reverse BFS order is computed with a simple BFS which takes $O(n+m)$ time.
    Then, for each $i$, we enumerate $\mathcal S_i = \mathcal S(G_{i-1} = G - \set{v_1, \ldots, v_{i-1}}, \{v_i\}, k)$.
    By the reversed BFS order, $G_{i-1}$ is a connected graph and thus will contain at least one solution (i.e., $k$-graphlet) as long as it has at least $k$ vertices, i.e., for $i<n-k$, which are the values of $i$ we consider.}
    
    \changed{We can thus enumerate all solutions in the set above by calling \enum{$G_{i-1}$,$\{v_i\}$,$k$}, which by \Cref{th:ksquare-graphlet-complexity} runs in $\order{k^2 \cdot N_i}$ time, where $N_i = |\mathcal{S}(G_i, \{v_i\}, k)|$}. 
    
    \changed{Finally, observe how the various $\mathcal S_i$ are a partition of the total set of solutions $\mathcal S(G, \emptyset, k)$: indeed for any graphlet $g$, let $v_j$ be the first vertex of $g$ in the reverse BFS order; by definition of $\mathcal{S}$ we have $g\in \mathcal S(G_{j-1}, \{v_j\}, k)$. Moreover, $g$ cannot be in any other $S_h$ as if $h<j$ the graphlets found contain $v_h\not\in g$, and if $h>j$ the graphlets found cannot contain $v_j$, as it is not in $G_{h-1}$. Thus $\bigcup\limits_{i\in \{1,\ldots, n-k+1\}}\mathcal{S}_i = \mathcal S(G, \emptyset, k)$, meaning \Cref{alg:graphlet_instances} will enumerate all $k$-graphlets of $G$ in $\order{n+m+ k^2 \cdot N}$ total time, i.e., $O(k^2)$ amortized time per solution.} 
    The space consumption of the whole strategy is $\order{n+m}$, needed by \Cref{alg:ksquare_graphlet} by \Cref{th:ksquare-graphlet-complexity}, as \Cref{alg:graphlet_instances} does not require asymptotically more space.
\end{proof}

\changed{We remark that, as shown in \Cref{th:bounded:degree}, the amortized cost is actually $O(k\cdot \min\set{k, \Delta})$, which technically could be lower than $O(k^2)$, although this is likely to be the case only in bounded-degree graphs.}

\section{Enumeration of edge \texorpdfstring{$k$}{k}-graphlets and \texorpdfstring{$k$}{k}-subtrees}\label{sec:es}

In this section we address the problem of enumerating all edge $k$-graphlets for a given graph $G$. We observe that we can transform this problem into the $k$-graphlet enumeration presented in Section~\ref{sec:graphlet} by replacing $G$ with its line graph (which represents the adjacency between the edges of $G$ rather than between its nodes) \changed{\cite{Harary1960LineGraph}}.
This gives an amortized $\order{k^2}$-time algorithm for enumerating all edge $k$-graphlets, but we \changed{are able to achieve} a better complexity \changed{for this task}, namely $\order{k}$ amortized time per solution and $\order{n + m}$ overall space. A slight variation of what we propose in this section also gives an enumeration algorithm for all $k$-subtrees with the same time and space complexity.

\changed{
For ease of reading, we summarize the key operations of the algorithm:
\begin{itemize}
    \item We use binary partition, where we select an edge $e$ and enumerate (i) all solutions containing $e$, then (ii) all solutions not containing $e$.
    \item We select $e$ so that the partial solution $\In$ chosen so far is connected.
    \item Every time we select or exclude an edge, we identify special types of edges (\textit{unnecessary}, \textit{mandatory}, \textit{far}) and process them with dedicated efficient strategies, that allow us to amortize the time complexity via push-out amortization. 
\end{itemize}
}

\subsection{Trimming \changed{mandatory and unnecessary edges}}\label{subsec:trimming}


In the following, given a graph $G = (V, E)$ and two disjoint sets $\In, \Ex \subseteq E$ of its edges, such that $G[\In]$ is connected, we denote by $\mathcal S(G, \In, \Ex, k)$ the set of edge $k$-graphlets \changed{in $G$} that contain $\In$ and do not contain $\Ex$. We denote by $\Vin$ the set of vertices $V(G[\In])$, \changed{and we recall that $V(G[\In])$ contains all and only the vertices that are incident to at least one edge in $\In$.}

\changed{As $\mathcal S(G, \In, \Ex, k) = \mathcal S(G - \Ex, \In, \emptyset, k)$, we can simply ignore $\Ex$ and hereafter use the reduced notation $S(G, \In, k)$, and denoting the exclusion of a set of edges $\Ex$ by simply removing them from $G$, i.e., $S(G - \Ex, \In, k)$}. Recall that the removal of $\Ex$ can be done in $O(\size{\Ex})$ time as remarked in Section \ref{subsec:bp}.

\changed{
Moreover, as edges from $G$ are removed, some edges may become \textit{mandatory}: an edge $e$ is mandatory if the connected component in $G - e$ containing $\Vin$ has less than $k$ edges (hence, $e$ must be part of \emph{all} solutions containing $G[\In]$).
In principle, we could immediately add all mandatory edges directly to $\In$, however, since we want to preserve $\In$ to be a connected subgraph, the algorithm will only add the mandatory edges adjacent to $\Vin$ (and iteratively add any mandatory edge that may become adjacent to $\Vin$ in this process).
while all mandatory edges must eventually be added to $\In$, it is not necessary to do this immediately, as by only recurring into instances that have at least one solution, we will always finally reach a point where all the mandatory edges are adjacent to $\Vin$ and are thus added to $\In$. Let us consider the complexity of this operation:
}

We also \changed{want the algorithm to} remove \emph{unnecessary} edges \changed{: these are the edges that cannot be part of solutions containing $\In$, as reaching them would require including more than $k$ edges. We define these as edges $e$ such that $dist(\Vin, e) \ge k - \size{\In}$, i.e., those whose distance from the closest vertex of $\Vin$ to an endpoint of $e$ is at least $k - \size{\In}$.}

\changed{We call the identification and removal of unnecessary edges from $G$, as well as the (iterative) addition to $\In$ of all mandatory edges adjacent to $\Vin$, the \emph{trimming operation}, and denote it by \trim{$G, \In, k$}.}

\changed{More formally, the operation \trim{$G, \In, k$} will return a pair $\langle H, \In' \rangle$ where $\In'$ is $\In$ plus the mandatory edges that were adjacent to $\Vin$, and $H$ is the graph obtained by removing the unnecessary edges from $G$, i.e., $E(H) = \{e\in E(G) | dist(V_{\In'}, e) \ge k - \size{\In'} \}$. We now show that this operation can be performed in $O(|E(H)|)$ time.}

\begin{lemma}\label{lem:trim:unnecessary} 
 \changed{Removing unnecessary edges in the instance $G, \In, k$ can be performed in $\order{\size{E(H)}}$ time, where $H$ is the graph obtained after removing such edges.}
\end{lemma}
\begin{proof}

    \changed{We show how to find the edges of $H$ in  $\order{\size{E(H)}}$ time, from which we can directly construct a copy of $H$. Firstly, all edges of $\In$ are part of $H$. Moreover, if $\size{\In}= k$ then there is no other edge and we are done; otherwise, as $\size{\In} \le k$, any edge neighboring $\In$ can potentially be added and thus will be part of $H$.}

    \changed{To find the remaining edges of $H$ let us contract all edges of $\In$: as contracting an edge $\{x,y\}$ can be done in either $\order{d(x)}$ or $\order{d(y)}$ (see Section~\ref{subsec:bp}), we can complete the operation in $\sum_{v\in \Vin}(d(v))$ time by choosing one arbitrary vertex $r\in \Vin$ and iteratively contracting every edge $\{r,x\} \in \In$ adjacent to $r$ in $\order{d(r)}$ time. Indeed since $\In$ is connected, this will finally contract all its edges, and by the above observation all edges incident to $\Vin$ are part of $H$, so $\sum_{v\in \Vin}(d(v)) \le \size{E(H)}$.}

    \changed{Let $r$ be the vertex obtained by contracting $\In$:
    We can now perform a breadth-first search from $r$ limited to distance $k - \size{\In}$, which will identify all edges that belong to $H$ without traversing any of the farther away ones that do not, i.e., we can identify all edges of $H$ in $\order{\size{E(H)}}$ time.}
\end{proof}

\changed{Now it is crucial to observe that, by definition, any edge that is mandatory cannot be also unnecessary. This means that the trimming operation can first remove unnecessary edges, obtaining $H$, and only then add the mandatory ones to $\In$. Moreover, adding the mandatory edges does not increase distances from $\Vin$, so it does not in turn create more unnecessary edges, and there is no need to reiterate the process. Thus, given $H$, we now focus on how to process the mandatory edges:}

\begin{lemma}\label{lem:trim:mandatory}
    \changed{Iterative identification and addition to $\In$ of all mandatory edges adjacent to $\Vin$, in the instance $H, \In, k$, can be performed in $O(|E(H)|)$ time.}
\end{lemma}
\begin{proof}
    \changed{
        Firstly, if $|E(H)| = k$ then all edges are mandatory and we are done. Otherwise, any mandatory edge $e$ must be a \textit{bridge}, i.e., an edge whose removal disconnects the graph, causing the loss of more than just $e$. In the following let us call the number of edges unreachable from $\In$ when removing $e$ the \textit{loss} of $e$.}

        \changed{We can identify all bridges in $O(|E(H)|)$ time \cite{DBLP:journals/ipl/Tarjan74}, and these actually give us a tree-like decomposition where nodes are 2-edge-connected component (2ECC), and they are connected by the bridges.}

        \changed{Choose as root of this tree an arbitrary 2ECC that contains at least one node of $\Vin$, and let us proceed in a bottom up way: for each leaf component $C$, its size is exactly the loss of the bridge connecting $C$ to the tree (plus one, for the bridge itself). Going up, each internal component will cause a loss equal to its own size, plus the size of all its children bridges, which we have already computed in the bottom-up procedure. Of course, once we find a bridge in $\In$ or a component containing edged of $\In$ we can stop in this subtree since they cannot be removed. Clearly, a bridge is mandatory if its loss is more than $|E(H)| - k$, thus we can immediately label every such bridge as mandatory. This process takes $O(|E(H)|)$ time as 2ECC are disjoint and the size of each is only computed once.}
        
        \changed{Finally, we go top-down from the components containing vertices in $\Vin$ and add to $\In$ all mandatory bridges adjacent to $\Vin$, as well as iteratively the mandatory bridges adjacent to them (as they became part of $\In$). This part also takes $O(|E(H)|)$ time as it can be done via a breadth-first search. This completes the proof.}
\end{proof}

\changed{From Lemmas \ref{lem:trim:unnecessary} and \ref{lem:trim:mandatory} we immediately obtain the following:}

\begin{corollary}\label{lem:tree:time:trim} 
 \changed{The trimming operation \trim{$G, \In, k$}, i.e., iterative addition of mandatory edges adjacent to $\Vin$, and removal of unnecessary edges, can be performed in $\order{\size{E(H)}}$ time, where $H$ is the graph obtained from $G$ after removing unnecessary edges. Note that  $|E(H)|\le |E(G)|$.}
\end{corollary}

For this reason, in the rest of this section we focus on the problem of enumerating the \changed{edge} $k$-graphlets of a generic instance $\mathcal S (G, \In, k)$, with the understanding that any excluded edges (i.e., $\Ex$) and unnecessary edges will be handled by the trimming operation. We say that an instance is \emph{trimmed} if it has no unnecessary edges and $\Ex = \emptyset$.

\subsection{\changed{Far edges}}\label{subsec:far_edges}

\changed{Given an instance of $G$, our partially formed solution $\In$ and $k$, our goal is to enumerate $\mathcal{S} (G, \In, k)$. However, to optimize the running time, we enumerate some of the solutions with a different strategy, using \textit{far edges}.}

\changed{We say that an edge $e$ is a \textit{far edge} if $dist(\Vin, e) = k - |\In| -1$, and define the set of all far edges as $E_{\tt far} = \inset{e \in E}{dist(\Vin, e) = k - |\In| -1}$. Intuitively, this means that the only way to find a solution including the edge $e$, is to include a shortest path from $\Vin$ to $e$. Thus we can transform the problem of enumerating the solutions that contain a far edge into a shortest-path enumeration problem, then remove $E_{\tt far}$ from $G$ to enumerate the rest.
}

Formally, we partition $\mathcal S(G, \In, k)$ into two disjoint sets, which are enumerated separately: 
\begin{itemize}
    \item $\mathcal S_{\tt far}(G, \In, k)$ is the set of solutions \changed{in $\mathcal S(G, \In, k)$ that contain a far edge.}
    \item $\mathcal S_{\tt near}(G, \In, k)$ \changed{are the remaining solutions, i.e.,}  $\mathcal S(G, \In, k) \setminus \mathcal S_{\tt far}(G, \In, k)$. 
\end{itemize}

\begin{observation}\label{lem:tree:sp}
    Let $H$ be an edge $k$-graphlet in $\mathcal S_{\tt far}(G, \In, k)$. Then $H - e$ must be a shortest path between $\Vin$ and an endpoint of $e$.
\end{observation}

From \Cref{lem:tree:sp}, we can enumerate all edge $k$-graphlets in \changed{$\mathcal S_{\tt far}(G, \In, k)$} using an algorithm for enumerating all shortest paths.
Enumeration of all shortest paths can be done in constant amortized time with linear time preprocessing (see~\cite[Section 4.1]{conte_et_al:LIPIcs.ISAAC.2023.21}).
\changed{The authors of \cite{conte_et_al:LIPIcs.ISAAC.2023.21} gave a constant amortized time enumeration of all $s$-$t$ paths on DAGs. 
On the other hand, our problem is multi-source and multi-sink shortest path enumeration, but it is known that this problem can be reduced to a $s$-$t$ path enumeration on DAGs \changed{(see, for instance, \cite{Birmele:Optimal:2013})}.}

\changed{As $\mathcal S_{\tt far}(G, \In, k)$ can be enumerated efficiently, we then focus on the enumeration of $\mathcal S_{\tt near}(G, \In, k)$. Firstly, observe that we can find $\mathcal S_{\tt near}(G, \In, k)$ by simply removing $E_{\tt far}$ from $G$:}

\begin{lemma}\label{lem:tree:near}
    $\mathcal S_{\tt near}(G, \In, k) = \mathcal S(G - E_{\tt far}, \In, k)$. 
\end{lemma}%
\begin{proof}%
    For any $H \in \mathcal S_{\tt near}(G, \In, k)$, $H$ cannot be in $\mathcal S_{\tt far}(G, \In, k)$  as it does not contain any edge from $E_{\tt far}$; hence, $H \in \mathcal S(G - E_{\tt far}, \In, k)$. Vice versa, for any $H \in \mathcal S(G - E_{\tt far}, \In, k)$, $H$ is contained in $\mathcal S_{\tt near}(G, \In, k)$ by definition.
\end{proof}

\changed{By Lemma~\ref{lem:tree:near}, it suffices to perform \trim{$G - E_{\tt far}, \In, k$} and enumerate all edge $k$-graphlets in the resulting trimmed instance.}

\subsection{\changed{Partition strategy}}\label{subsec:partition_strategy}
\Cref{alg:k-subtrees} summarizes the partition strategy, where we organize the computation so that the push-out amortization can be applied.
\changed{Recall that we select $e$ so that $\In$ stays connected, and that we only immediately add to $\In$ the mandatory edges adjacent to $\Vin$.}

\changed{A key idea to make the amortization work are \textit{heavy edges}:} we say that an edge $e$ is \emph{heavy} if the number of unnecessary edges in $G - e$ is \changed{at least $\size{E}/2$}.

\begin{lemma}\label{lem:heavy}
    Let $G, \In, k$ be a trimmed instance. \changed{$\Gamma(\Vin)$ has at most one heavy edge, and least one non-heavy edge.}
\end{lemma}
\begin{proof}
    \changed{Firstly, since the instance is trimmed, $\Gamma(\Vin)$ has no mandatory edges, and has at least 2 edges (if there was only one edge, it would be mandatory to use it). 
    Observe also how $G$ has initially no unnecessary edges, so the unnecessary edges that determine whether an edge $e$ is heavy or not are only the ones created by the removal of $e$ itself.}
    
    \changed{Suppose that $\Gamma(\Vin)$ has two heavy edges $e \neq f$. Let $U_e$ be the unnecessary edges in $G-\{e\}$, and $U_f$ be the unnecessary edges in $G-\{f\}$. We first prove that $U_e$ and $U_f$ are disjoint: if an edge $g$ is in $U_e$, since it was not unnecessary in $G$, it means in $G$ it had a shortest path towards $\Vin$ that used $e$; this shortest path clearly cannot also use $f$, as just one between the edges $e$ and $f$ is sufficient to connect to $\Vin$. Thus, this shortest path is not broken by the removal of $f$, meaning $g\not\in U_f$. This means $U_e \cap U_f = \emptyset$. Moreover, $e$ and $f$ are directly adjacent to $\Vin$ so $e\not\in U_f$ and $f\not\in U_e$.} 
    
    \changed{Let $m = |E(G)|$. We have $|U_e|+|U_f| \le m-2$ as $e$ and $f$ themselves do not belong to these sets; this contradicts that both $e$ and $f$ can be heavy as it would imply $|U_e|+|U_f| \ge m/2 + m/2 = m$. Thus there is at most one heavy edge. Finally, as we observed $\Gamma(\Vin)$ has at least 2 edges, it follows that at least one is non heavy.}
\end{proof}


\changed{The following theorem establishes a bound on the space consumption of our strategy.}

\begin{theorem}\label{theo:cor}
    \Cref{alg:k-subtrees} \changed{computes} $\mathcal S(G, \In, k)$ using $\order{n + m}$ space.
\end{theorem}
\begin{proof}
    By \Cref{lem:tree:sp} and \Cref{lem:tree:near}, we can enumerate all solutions in $\mathcal S_{\tt far}(G, \In, k)$ and 
    $\mathcal S_{\tt near}(G, \In, k) = \mathcal S(G - E_{\tt far}, \In, k)$\changed{, respectively}.
    By \Cref{lem:heavy}, $\Gamma(\Vin)$ contains at least one non-heavy edge $e$.
    Moreover, when $\Gamma(\Vin)$ has only one edge, we add this edge to $\In$ until $\Gamma(\Vin)$ \changed{contains} at least two edges \changed{(i.e. we follow the chain of mandatory edges)}.
    Thus, we can partition $\mathcal S(G - E_{\tt far}, In, k)$ into two sets using a non-heavy edge.
    Moreover, in each node, we store only the set of edges added to $\In$ and the set of removed edges.
    Therefore, the total space requirement is $\order{n + m}$ at any time.
\end{proof}

\DontPrintSemicolon
\begin{algorithm}[t]
    \caption{Given a graph $G$, a set $\In$ of its edges inducing a connected subgraph, and a non-negative integer $k$, the algorithm enumerates all edge $k$-graphlets  that contain $\In$.}
    \label{alg:k-subtrees}
    \KwIn{Graph $G = (V, E)$, edge set $\In \subseteq E$, non-negative integer $k$}
    \KwOut{All edge $k$-graphlets of $G$ that contains $\In$}
    \SetKwProg{Fn}{Function}{:}{end}
    
    \Fn{\EnumEdge{$G = (V, E), \In, k$}}{
        \lIf{$\size{\In} = k$}{\label{alg:k-subtrees:line:if-k}
            Output $\In$ and \Return
        }
        \If{$\size{\In} = k - 1$}{
            \lFor{$f \in \Gamma(\Vin)$}{\label{alg:k-subtrees:line:for-k-1}
                \textbf{output} $\In \cup\set{f}$;
                \Return
              }
        }            
        
        Enumerate all solutions that contain a far edge. \label{alg:shortest-paths} \tcp*{See Section \ref{subsec:far_edges}}

        $G \gets G - E_{\tt far}$ \tcp*{Remove all far edges}
        
        Let $e$ be a non-heavy edge in $\Gamma_{G}(\Vin)$\label{alg:e-pick}\tcp*{$e$ exists by \Cref{lem:heavy}}
        \BlankLine
        \changed{$\langle H', \In' \rangle \gets \trim{G, \In \cup \set{e}, k}$}\;
        \changed{\EnumEdge{$H', \In', k$ }}\label{alg:trim_in}\tcp*{Solutions containing $e$}
        \BlankLine
        \changed{$\langle H'', \In'' \rangle \gets \trim{G - \{e\}, \In, k}$}\;
        \changed{\EnumEdge{$H'', \In'', k$}}\label{alg:trim_ex}\tcp*{Solutions not containing $e$}
    }

    \Fn{$\trim{G, \In, k}$}{
        \While{$\Gamma(\Vin)$ contains a mandatory edge}{
            Let $e$ be a mandatory edge in $\Gamma(\Vin)$\;
            $\In \gets \In \cup \{e\}$\;
        }   

        $E_H \gets \{e\in E(G) | dist(\Vin, e) < k - |\In|\}$
    }
\end{algorithm}

\subsection{Time complexity analysis}
We now show how to apply the push-out (PO) amortization to state the time complexity analysis of \Cref{alg:k-subtrees}.
We first show that, \changed{given a recursive node $X$ corresponding to \EnumEdge{$G = (V, E), \In, k$}}, the algorithm runs \changed{in amortized $\order{m}$ time, i.e., $\order{m\cdot N}$ time} if the instance received at the root is already trimmed, where $m = |E(G)|$. \changed{We will later improve this to $O(k)$ amortized time but this is a necessary first step.}  

\changed{By Corollary \ref{lem:tree:time:trim} and what stated in Section \ref{subsec:far_edges}, The trimming operation and the enumeration of shortest path \Cref{alg:shortest-paths} can be performed in $\order{m +\size{ch(X)}}$ time, where $\size{ch(X)}$ is the number of recursive children of $X$ (which is 2 for the binary partition, plus one for each shortest path enumerated on \Cref{alg:shortest-paths}.})

\changed{As the recursive subtree is a tree where every recursive node has at least 2 children, i.e., the total number of recursive nodes is at least half of the total number of leaves, and each leaf corresponds to one of the $N$ solutions, the total cost is $\order{m \cdot N}$.}

\changed{Now, our goal is to further improve the analysis, and show that \Cref{alg:k-subtrees} runs in amortized $\order{k}$ time.
Each leaf node demands $\order{k}$ time since each leaf node just outputs a solution. Let us formalize the cost of internal nodes:}

\begin{lemma}\label{lem:X:cost}  
\changed{The time required a recursive call $X=$ \EnumEdge{$G, \In, k$}, ignoring the cost if its two recursive children, is $\order{m_X + \size{ch(X)}}$,
where $m_X = |E(G)|$ and $\size{ch(X)}$ is the total number of children (including the solutions generated by enumerating shortest paths to far edges).}
\end{lemma}
\begin{proof}
\changed{Firstly, if triggered, Lines \ref{alg:k-subtrees:line:if-k}-\ref{alg:k-subtrees:line:for-k-1}, take just $O(m_X)$ time.
Then, \Cref{alg:shortest-paths} takes $\order{m_X + \size{ch(X)}}$ time as discussed in \Cref{subsec:far_edges}. Deleting $\Efar$ also takes $O(m_X)$ time as edges take constant time each to delete. Finally, the two trimming operations take $O(|E(H)|)$ time each by \Cref{lem:tree:time:trim}, where $H$ is the resulting graph generated, which is a subgraph of $G$, so $|E(H)|\le m_X$.}
\end{proof}

We will then show that each node $X$ satisfies the PO condition by setting $\alpha = \beta = 5/4$. \changed{To do so, we define $c^*$ as} an arbitrary constant that is larger than the constant hidden in the \changed{big O} notation for computation time at internal nodes and at leaf nodes.
\changed{Thus,} we assume that each internal node $X$ takes $c^*(m_X + \size{ch(X)})$ time and each leaf node $T$ takes $c^*\cdot k$ time.

\changed{It is worth noting that $O(k)$ amortized time is trivially obtained on instances where the graph has $O(k)$ edges:}
\begin{observation}\label{obs:4k_edges}
    \changed{When $m_X = O(k)$, \Cref{alg:k-subtrees} runs in amortized $\order{k}$ time since, as shown above it achieves amortized $\order{m}$ time where $m$ is the size of the graph, and in this case the size of the graph received by the recursive call $X$ is indeed $m_X\le 4k$ and thus $O(k)$, and in this case we have $m_X= \order{k}$.}
\end{observation}

\changed{Thus, although we do not actually use the assumption, in we observe that the amortization is only really necessary when $m_X = \omega(k)$.}

\changed{Next, let us consider lower bounds for the size of each trimmed instance received by a child during the recursion, analyzing the various operations that change the instance.}
\changed{We first consider how trimming after removal of far edges does not remove further edges:}

\begin{lemma}\label{lem:trim:far}
\changed{Let $G,\In\neq \emptyset,k$ be a trimmed instance, and $\Efar$ its far edges. Then, \trim{$G - \Efar, \In, k$} has at least $\size{E(G)} - \size{E_{\tt far}}$ edges.}
\end{lemma}
\begin{proof}
    From the definition of the operation \trim{}, we remove all unnecessary edges.
    \changed{Observe how a far edge could only be reached by $\In$ by adding a shortest path to $e$ and $e$ itself; any edge farther away from $\In$ could not be added as it would be unnecessary (thus not present in the trimmed instance). This means that no edge would require $e$ in order to reach $\In$, and thus the removal of the far edge $e$ does not cause any other edge to become unnecessary (or far), and thus nothing other than $\Efar$ is removed by \trim{$G - \Efar, \In, k$}.}
\end{proof}

\changed{Next, we show that adding one edge to $\In$ does not change the size of a graph when an instance is trimmed:}

\begin{lemma}\label{lem:add:trim}
    \changed{Let $G,\In, k$ be a trimmed instance with no far edges, i.e., $E_{\tt far} = \emptyset$. 
    For an edge $e \in \Gamma(\In)$, the graph obtained by \trim{$G, \In \cup \set{e}, k$} has $\size{E(G)}$ edges.}
\end{lemma}
\begin{proof}
    \changed{As \trim{} removes only unnecessary edges, it is sufficient to show that no edge becomes unnecessary: indeed, an edge is unnecessary if its distance from $\Vin$ is at least $|\In|-k$. However, every edge that is not a far edge, has distance from $\In$ smaller than $|\In|-k-1$, i.e., at most $|\In|-k-2$ . Adding $e$ increase $|\In|$ by one, thus the new distance from $\In \cup \{e\}$, in the worst case, will be at most $|\In|-k-2+1 = |\In|-k-1$. In other words, the addition of $e$ may create far edges, but not unnecessary edges.}
\end{proof}

\changed{Finally, we observe the impact of \trim{} when removing a non-heavy edge in a previously trimmed instance, which follows immediately from the definition of heavy edge.}

\begin{observation}\label{lem:heavy:trim}
    \changed{Let $G, \In, k$ be a trimmed instance, and let $e\in \Gamma(\Vin)$ be a non-heavy edge. Then, the graph obtained by \trim{$G - \{e\}, \In, k$} has at least $\size{E(G)}/2$ edges.}
\end{observation}

Now, we are ready to give our analysis.

\begin{lemma}   
    Let $X$ be a node in the recursion tree of \Cref{alg:k-subtrees}.
    Moreover, let $(G_X, \In_X, k_X)$ be an instance corresponding to $X$.
    Then, $X$ satisfies the PO condition \changed{(Theorem \ref{thm:push-out})} by setting $\alpha = \beta = 5/4$ and $T^*=O(k)$. 
\end{lemma}
\begin{proof}

    \changed{Let $F_X$ be the set of far edges in $G_X$.
    The number of children of $X$, i.e., $\size{ch(X)}$, is at least $\size{F_X}+2$ (at least one per far edge, plus two nested recursive calls).
    The graph obtained by \trim{$G - F_X, \In, k$} has at least $\size{E_X} - \size{F_X}$ edges by \Cref{lem:trim:far}, and this gets passed to the children recursive calls. \Cref{lem:heavy} tells us that there is always a non-heavy edge $e$ that we can consider, thus we can give lower bounds to the size of the graph received by the children recursive calls. The child on \Cref{alg:trim_in} adds $e$ so by \Cref{lem:add:trim}  has size $\size{E_X} - \size{F_X}$. The child on \Cref{alg:trim_ex} removes $e$ so by \Cref{lem:heavy:trim} has size $\ge (\size{E_X} - \size{F_X})/2$, each child derived from a far edge is a leaf, which takes $O(k)$ time to simply output their solution, and thus we have at least $\size{F_X}$ children of size $O(k)$ each. Recalling as noted above that the constant factors in the big $O$ notation are denoted by $c^*$, the cost of each recursive call given by \Cref{lem:X:cost}, and that $k\ge 2$ (as $k=1$ is a trivial problem), putting all together we have that the sum of running time of the children $ch(X)$ of $X$ is:}

    \changed{\begin{align*}
        \sum_{Y \in ch(X)}T(Y) & \ge c^*(\size{E_X} - \size{F_X}) + c^*(\size{E_X} - \size{F_X})/2 + c^*\cdot k \cdot \size{F_X}\\
        & = \frac{3}{2}c^*\size{E_X} - \frac{3}{2}c^*\size{F_X} + c^*\cdot k \cdot \size{F_X} \\
        & = \frac{3}{2}c^*\size{E_X} +  c^*(k -\frac{3}{2}) \size{F_X} \\
        & \ge \frac{3}{2}c^*\size{E_X} \\
    \end{align*}}
    
    \changed{We next analyze the term $\alpha T(X) - \beta(\size{ch(X)} + 1)T^*$ as follows. Recall that the number of children of $X$ obtained by enumeration of far edges is at least $|F_X|$, but could be exponentially more, let this number be $|F_X| + \gamma$, with $\gamma\ge 0$, and recall that $X$ has two more children obtained by recursive calls, finally, recall that the cost of a leaf is assumed to be $T^* = O(k) = c^*k$.}
    \changed{\begin{align*}
        & \alpha T(X) - \beta(\size{ch(X)} + 1)T^* \\
        & = \alpha c^* ( |E_X| + |F_X| + \gamma ) - \beta( 2 + |F_X|+ \gamma +1  ) c^*k\\
        & = \frac{5}{4} c^* ( |E_X| + |F_X| - k|F_X| + \gamma - k \gamma -3 )\\
        & = \frac{5}{4} c^* ( |E_X| - (k-1)|F_X| - (k-1)\gamma -3 )\\
        & \le \frac{5}{4} c^*\size{E_X}
    \end{align*}}
    \changed{Where for simplicity we ignored negative factors in the last step. 
    Therefore, we have 
    \begin{align*}
        \sum_{Y \in ch(X)}T(Y) \ge \frac{3}{2}c^* \size{E_X} > \frac{5}{4} c^*\size{E_X} \ge \alpha T(X) - \beta(\size{ch(X)} + 1)T^*
    \end{align*}
    We can thus conclude that the PO condition holds for every recursive node of \Cref{alg:k-subtrees}, with $\alpha = \beta = 5/4$ and $T^*=O(k)$.}
\end{proof}

Therefore, each internal node \changed{satisfies the PO condition} and we obtain the following result. 

\begin{theorem}\label{th:push-out-edge-graphlets}
    \Cref{alg:k-subtrees} runs in total $\order{n + m + k \cdot N}$ time and $\order{n + m}$ space \changed{if $\In \neq \emptyset$}, where $N$ is the number of solutions.
    If the instance received at the root of the recursion tree is already trimmed, \Cref{alg:k-subtrees} runs in $\order{k\cdot N}$ time.
\end{theorem}

\subsection{\changed{Generating all edge \texorpdfstring{$k$}{k}-graphlets using  \texorpdfstring{Algorithm~\ref{alg:k-subtrees}}{Algorithm 4}}}
\changed{As we did for the vertex $k$-graphlets, we now combine everything together to show that we can enumerate all edge $k$-graphlets in $O(n+m+k\cdot N)$ total time. 
To this end, we show our instance generation algorithm in \Cref{alg:generating_instances}. It essentially iterates over all edges, using each edge $e_i$ as the first building block of a potential solution; however, as $\EnumEdge$ assumes the input graph to be trimmed, it also performs the trimming operation before calling it. Finally, note how we consider the edges in a \textit{reversed BFS order of E}: this means performing a BFS from an arbitrary starting vertex $v$, writing the edges down in the order in which they are found, and finally reversing the sequence. This order has the significant advantage that every \textit{suffix} of it corresponds to a connected graph (as all remaining edges $e_i,\ldots,e_m$ are still able to reach $v$). This means that every instance generated will have at least $k$ edges reachable from $e_i$, except for when less than $k$ edges are left in the graph (which is why we stop at $i=m-k+1$).
}
By combining \Cref{alg:k-subtrees,alg:generating_instances}, we obtain the following theorem.

\begin{algorithm}[tb]
\DontPrintSemicolon
    \caption{\changed{Generate all edge $k$-graphlets of a graph $G$ using \Cref{alg:k-subtrees} as subroutine.}}
    \label{alg:generating_instances}
    \KwIn{A connected graph $G = (V, E)$, a positive integer $k$}
    \KwOut{all edge $k$-graphlets.}
    \SetKwFunction{GenIns}{Gen-Ins}{}{}
    \Fn{\GenIns{$G = (V, E), k$}}{
        \changed{Let $e_1, \ldots, e_m$ be a reversed BFS order of $E$} \;
        $G_0 \gets G$\;
        \For(\tcp*[h]{We can skip the final $k-1$ edges}){\changed{$i = 1, 2, \ldots, (m-k+1)$}}{ 
            \changed{$\langle H, \In\rangle \gets$ \trim{$G_{i-1}$,$\{e_i\}$,$k$}}\;
            \changed{\EnumEdge{$H,\In,k$}}\;
            \changed{$G_i \gets G_{i-1} - \set{e_i}$}\label{line:ei:removal}\;                 
        }
    }
\end{algorithm}

\begin{theorem}\label{thm:main}
    Let $G$ be a connected graph and $k$ be a positive integer. 
    We can enumerate all edge $k$-graphlets in $\order{n + m + k\cdot N}$ time and $\order{n + m}$ space, where $N$ is the number of solutions.
\end{theorem}
\begin{proof}
    \changed{The proof follows from the same argument used in Theorem~\ref{th:wrapper_graphlets}. 
    As computing the reverse BFS order requires performing a BFS, which takes $O(n+m)$ time just once.
    The removal of $e_i$ on \Cref{line:ei:removal} takes constant time, so the total contribution of this line is $O(m)$ time. Finally, we need to call \trim{$G_{i-1}$,$\{e_i\}$,$k$}, which produces the pair $\langle H, \In\rangle$ in $O(|E(H)|)$ time by \Cref{lem:tree:time:trim}: we can actually just charge this cost onto the recursive call $X =$ \EnumEdge{$H,\In,k$} generated, as by \Cref{lem:X:cost}, $X$ already takes $O(|E(H)|+|ch(X)|)$ time, so this does not alter its complexity. 
    Finally, Theorem~\ref{th:push-out-edge-graphlets} shows that \EnumEdge{} respects the PO condition with $T^*=O(k)$, which by \Cref{thm:push-out} means that it achieves $O(k)$ amortized time per graphlet enumerated. The running time is thus $O(n+m+k\cdot N)$ total time, where $N$ is the number of edge $k$-graphlets contained in $G$.}
\end{proof}

\subsection{Modifications for \texorpdfstring{$k$}{k}-subtrees enumeration}
\changed{In this section, we propose an algorithm for enumerating all $k$-subtrees in a graph, obtained by small modifications of \Cref{alg:k-subtrees}.}
\changed{Let $X$ be a node of the recursion tree as seen in the previous sections, $G_X$ be the associated graph, $\In_X$ be a set of edges, and $k$ be an integer.}
\changed{The main difference between this algorithm, and \Cref{alg:k-subtrees} lies in the partition strategy.
Indeed, for an edge $e = \set{u, v} \in \Gamma_{G_X}(\In_X)$ such that $u$ is incident to any edge of $\In_X$, 
an edge that connects vertices $v$ and a vertex in $G_X[\In_X]$ cannot be added to $\In_X \cup \set{e}$ (otherwise it would form a cycle), hence we remove such edges from $G_X$.
The correctness of this algorithm can be shown using the same arguments adopted in Section~\ref{sec:es}.
We can also prove that the PO condition holds for this algorithm: in particular, when we add an edge $e = \{u,v\}$ to $\In_X$, and there exist many edges that connect $v$ and another vertex in $G_X[\In_X]$, we can generate all subproblems (i.e. children nodes in the recursion tree) in linear time.
Therefore, we can enumerate also $k$-subtrees in $O(k)$ amortized time, as stated in the following theorem.
}
%
\begin{theorem}
    Let $G$ be a graph and $k$ be an integer. 
    We can enumerate all $k$-subtrees of $G$ in $\order{n + m + k\cdot N}$ time and $\order{n + m}$ space, where $N$ is the number of $k$-subtrees.
\end{theorem}

\section{Conclusions and Future Work}
We provided new algorithms for enumerating both vertex-induced and edge-induced subgraphs/subtrees in any simple, undirected graph, improving the state of the art for both problems.
This improvement over the fastest known algorithms in the literature was possible by the use of both traditional amortized analysis and the more recent push-out amortized analysis.
Our algorithms are the first to have time complexity depending \emph{solely} on the size of the subgraph to be enumerated, \changed{i.e., $O(k^2)$ amortized time for $k$-graphlets and $O(k)$ amortized time for edge $k$-graphlets,} in contrast with existing approaches that also include the size of the graph or its maximum degree in their complexity bounds.

The question on whether $O(k)$ amortized time is attainable for $k$-graphlets is still open, in contrast to what we did with edge $k$-graphlets/subtrees, because a $k$-graphlet can have a quadratic number of edges, making this open problem challenging.

\subsection*{Acknowledgments}
\changed{We wish to thank the anonymous reviewers for the detailed feedback that allowed us improve the clarity of the paper.
The work was partially supported by MUR PRIN 2022 project EXPAND: scalable algorithms for EXPloratory Analyses of heterogeneous and dynamic Networked Data (\#2022TS4Y3N). 
}

\bibliography{kurita,bib2}

\end{document}